\documentclass[a4paper,10pt]{llncs}

\usepackage{ucs}
\usepackage[utf8x]{inputenc}
\usepackage{xcolor}
\usepackage{fontenc}
\usepackage[ruled,vlined,linesnumbered]{algorithm2e}
\usepackage{hyperref}
\hypersetup{linkbordercolor=1 1 1}
\usepackage[dvips,pdftex]{graphicx}
\usepackage{ multirow, longtable, tocbibind, amssymb, enumerate}
\usepackage{enumerate}

\newcommand{\Aut}{\mathrm{Aut}}
\newcommand{\Sym}{\mathrm{Sym}}

 \newtheorem{test}[theorem]{Test}

\title{Ternary graph isomorphism in polynomial time, after Luks}
\author{Adri\`{a} Alcal\'{a} Mena , Francesc Rossell\'o}
\institute{Dept. Applied Mathematics and Computer Science , Univ. of Cantabria , 39005 Santander, Spain
\\
Dept. Mathematics and Computer Science, Univ. of the Balearic Islands, 07122 Palma de Mallorca, Spain}

\begin{document}
\maketitle
\begin{abstract}
The  graph isomorphism problem has a long history in mathematics and computer science, with applications in computational chemistry and biology, and it is believed  to be neither solvable in polynomial time nor NP-complete. E. Luks proposed in 1982 the best  algorithm so far for the solution of this problem, which moreover runs in polynomial time if an upper bound for the degrees of the nodes in the graphs is taken as a constant. Unfortunately, Luks' algorithm is purely theoretical, very difficult to use in practice, and, in particular, we have not been able to find any implementation of it in the literature. 
The main goal  of this paper is to present an efficient implementation of this algorithm for ternary graphs in the SAGE system, as well as an adaptation to  fully resolved rooted phylogenetic networks on a given set of taxa. 
\end{abstract}
\section{Introduction}

The \emph{graph isomorphism problem} ---the  problem of deciding whether two finite graphs are isomorphic or not--- has a long history in mathematics and computer science. It is  one of the most important decision problems for which the computational complexity is not known yet~\cite{goldberg03,KST93}, and it is believed to be neither in P nor NP-complete.  Its uncertain status has led to the definition and study of the class GI of  the decision problems  that are polynomial-time Turing reducible  to it \cite{Hoff82}.

The best current theoretical algorithm for the solution of the graph isomorphism problem is due to Eugene Luks \cite{Luks81,Luks83}. His algorithm runs in time $2^{O(\sqrt{n}\log(n))}$, where $n$ is the number of  nodes in the graphs, and it relies on the  classification of finite simple groups \cite{FSG}. Without using this major  result,
whose proof consists of tens of thousands of pages in several hundred journal articles  and for which a unified, simplified and revised version is still in progress,  Babai and Luks  \cite{BabLuks83} gave an algorithm that runs in $2^{O(\sqrt{n}\log(n)^2)}$ time. 
   
The graph isomorphism problem has many practical applications outside mathematics. Our interest in it stems from its applications in computational biology and bioinformatics. Graphs are ubiquitous in biology as models of different complex systems: molecular structures, phylogenetic trees and networks, metabolical pathways, protein-protein interaction (PPI) networks, gene expression networks, etc. \cite{AS,Hub,Huson,MV}. In all these fields, the comparison of  graphs is an important computational problem. For instance, 
a protein's function is closely related to its three-dimensional shape, which at an intermediate level of detail is  modeled by a contact graph, and thus the comparison of such contact graphs is a key tool in the prediction of  proteins' function \cite{GI,TA}; the comparison of phylogenies is acknowledged to be one of the main problems in phylogenetics \cite{fel:04,Huson,steelpenny:sb93};  and the comparison of metabolic pathways and PPI networks is currently one of the hot tools in the study of evolution \cite{FS,KK}. The most basic graph comparison problem  in all these contexts is, of course, the  detection of equalities. Without being able to decide efficiently whether two such graphs are isomorphic or not, we cannot expect to define metrics, matchings or alignments that can be computed efficiently. 

So,  it is important to design efficient algorithms to test if two graphs are isomorphic, at least for some special classes of graphs with practical applications. The aforementioned algorithm by Luks \cite{Luks83} is fixed-parameter tractable, in the sense that it runs in polynomial time if we consider an upper bound $d$ for the degrees of the graphs under comparison as a constant. Such an algorithm, even for \emph{ternary graphs} (graphs with all their nodes of degree at most 3), has relevant applications in computational biology: for instance, in phylogenetics, in the comparison of arbitrary  fully resolved rooted phylogenetic networks, and in the comparison of split networks \cite{Huson}; in both types of  networks, all nodes have total degrees at most 3.  For these networks, no specific polynomial-time isomorphism test has been devised yet, and the authors working with them simply quote Luks' result when they need to state that their isomorphism can be decided in polynomial time; see, for instance, \cite[
p. 168]{Huson} or \cite[\S V]{Nak}. 

The problem with Luks' algorithm  is that it is purely theoretical, very difficult to use in practice, and, in particular, we have not been able to find any implementation of it in the literature.  The main goal of this paper is, then, to present an efficient implementation of this algorithm in the SAGE system. We also present an implementation of an adaptation of this algorithm that solves the isomorphism problem for fully resolved rooted phylogenetic networks.
This paper is based on the first author's Master Thesis  \cite{MasterThesis}, where we refer the reader for more details.

\section{Luks' algorithm for  ternary graphs}

In this section we explain Luks' algorithm for \emph{ternary} graphs, that is, for connected graphs with all their nodes of degree at most $3$.  We shall omit all proofs, which can be found in the original paper by Luks \cite{Luks83} or, with more detail, in the first author's Master Thesis  \cite{MasterThesis}.
 
We begin by reducing the graph isomorphism problem to a  problem about isomorphisms of groups.  Given two ternary graphs $X_1=(V_1,E_1)$ and $X_2=(V_2,E_2)$, and two edges $e_1\in E_1$ and $e_2\in E_2$, consider the ternary graph 
$X=\mathtt{BuildX}(X_1, X_2, e_1, e_2)$ obtained by splitting $e_1$ and $e_2$ by means of new nodes $v_1$ and $v_2$, respectively, and then connecting  these new nodes $v_1$ and $v_2$ by means of a new edge $e$. Then, $X_1$ and $X_2$ are isomorphic if, and only if, given any $e_1\in E_1$, there exist some $e_2\in E_2$ and some automorphism $\sigma$ of $X$ such that $\sigma(v_1)=v_2$, and in particular such that $\sigma(e)=e$.
And then, if any such automorphism of $X$ does exist, then any set of generators of $\mathrm{Aut}_e(X)$, the group of automorphisms of $G$ that fixes $e$, will contain some.
This provides the following algorithm.

\begin{algorithm}\label{IsotoAute}\hypertarget{IsotoAute}{}
\KwData{$X_1, X_2$ ternary graphs}
\KwResult{Test if $X_1$ and $X_2$ are isomorphic}
\SetKwFunction{BuildX}{BuildX}
\Begin{
$e_1 \in E(X_1)$ \\
\For{$e_2 \in E(X_2)$}{
	$X \leftarrow $ \BuildX{$X_1, X_2, e_1, e_2 $}\\
	$S\leftarrow$ a set of generators of $\mathrm{Aut}_e(X)$\\
	\For{ $\sigma \in S $}{
	    \If{ $\sigma(v_1) == v_2 $}{
		\Return{ True}
	    }
	}
    }

\Return{False}
}
\caption{Isomorphism of ternary graphs}

\end{algorithm}

So, we are reduced to compute a set of generators of $\Aut_e(X)$, for 
a ternary graph $X$ and any $e\in E(X)$. This group is determined through a natural sequence of successive ``approximations'' $\Aut_e(X_r)$, where  each $X_r$ is the 
subgraph consisting of all nodes and all edges of $X$ which appear in paths of length $\leq r$ containing $e$. More formally, if $e= (a,b)$, let
$$
X_1=(\{a,b\},\{(a,b)\})
$$ 
and, for every $r\geq 2$,
$$
\begin{array}{l}
V ( X_r) = \{ y \in V( X)\mid \exists\, x \in V(X_{r-1} ) \mbox{ such that } (x,y) \in E(X) \} \\
E(X_r) =  \{ ( x,y) \in E( X) \mid   \exists\, x \in V(X_{r-1} ) \mbox{ such that } (x,y) \in E(X) \}
\end{array}
$$
There exist  natural homomorphisms
$$
 \pi_r : \Aut_e(X_{r+1}) \rightarrow \Aut_e(X_r)
$$
defined by the restriction of the automorphisms. These homomorphisms allow us to construct  recursively a generating set for $\Aut_e(X_{r+1})$ from one for $\Aut_e(X_r)$ by solving, for every $r\geq 1$,  the following two problems:
\begin{enumerate}
\item[(I)] Compute a set $\mathcal{K}_{r}$ of generators of $\ker \pi_r$.
\item[(II)] Compute a set $\mathcal{S}_{r+1}$ of generators of $\pi_r(\Aut_e(X_{r+1}))$. 
\end{enumerate}
Indeed, if $\mathcal{S}_{r+1}'$ is any set of pre-images  of $\mathcal{S}_{r+1}$ in $\Aut_e(X_{r+1})$, then  $\mathcal{K}_{r} \cup \mathcal{S}_{r+1}'$ generates
 $\Aut_e(X_{r+1})$.

So, if we know how to solve  Problems I and II,  the algorithm to compute a set of generators of $\Aut_e(X)$ is the following. In it, and in the sequel, given two elements $a,b$ of a set $Y$, we denote by $(a\; b)\in \Sym(Y)$ the transposition $a\leftrightarrow b$.

\begin{center}
 
\begin{algorithm}\label{Aute}\hypertarget{Aute}{}
\KwData{A ternary graph $X$ and an edge $e$}
\KwResult{A set of generators of $\Aut_e(X)$}
\Begin{
Compute the sequence of subgraphs $X_1\subseteq \cdots \subseteq X_N=X$ as before\\
$\Aut_e = \{(a\; b)\}$\\
\While{$r< N$}{
	$K \leftarrow$ a set of generators of $\ker \pi_r$ \\
	$S \leftarrow $ a set of generators of $\pi_r(\Aut_e(X_{r+1}))$ \\
	$S' \leftarrow $ a set of pre-images of $S$ in $\Aut_e(X_{r+1})$\\
	$\Aut_e = S' \cup K$\\
	$r=r+1$
    }
\Return{$\Aut_e$}
}
\caption{A set of generators of $\Aut_e(X)$}

\end{algorithm}
 \end{center}

Set now  $V_{r+1} = V(X_{r+1}) \setminus V(X_{r})$ and $A_r=\{a\subseteq V_r\mid 0<|a|\leq 3\}$, and consider the mapping
$$
 f: V_{r+1} \rightarrow A_r
$$
defined by $f(v)= \{ w \in V(X_r) \mid (v,w) \in E(X) \}$; we call $f(v)$ the \emph{neighbor set} of $v$. 
The following result solves Problem I.

\begin{proposition}\label{K}
$\mathcal{K}_{r}=\{(u\; v)\in \Sym (V_{r+1})\mid u\neq v,\ f(u)=f(v) \}.$
\end{proposition}
Moreover, the fact that each $\ker \pi_r$ is generated by transpositions implies, by induction, the following result, which will be useful later.

\begin{proposition}\label{Tutte}
 For each $r$, $\Aut_e(X_r)$ is a $2$-group.
\end{proposition}

So, our graph isomorphism problem is reduced to solve Problem II. Now, consider the following  three subsets of $A_r$:
$$
\begin{array}{l}
A_r' = \big\{ \{v_1, v_2\} \in A_r\mid  ( v_1, v_2 ) \in  E(X_{r+1}) \big\}\\
A_{r,1} = \big\{ a \in A_r\mid  a = f(v) \mbox{ for some unique } v \in V_{r+1} \big\}\\
A_{r,2} = \big\{ a \in A_r\mid  a=f(v_1) = f(v_2 ) \mbox{ for some } v_1  \neq v_2 \big\}
\end{array}
$$
We have the following result.

\begin{proposition}
$\pi_r(\Aut_e(X_{r+1}))$ is the set of those $\sigma \in \Aut_e (X_r)$ that stabilize the sets $A_{r,1}$,  $A_{r,2}$, and $A_r'$.
\end{proposition}

Now, set $B_r = V(X_{r-1}) \cup A_r$ and   extend the action of $\Aut_e(X_r)$ on $V(X_r)$ to $B_r$ in the natural way: if  $a \in A_r$, then $\sigma(a) = \{ \sigma(w)\mid  w \in a \}$.  Color each element of $B_r$ with one of five colors that distinguish, on the one hand, whether or not it is in $A_r'$, and, on the other hand, whether it is in $A_{r,1}$, or $A_{r,2}$, or neither. Only five colors are needed, since $A_r'\cap A_{r,2}=\emptyset$.

By the previous proposition,  $\sigma \in \pi_r ( \Aut_e(X_{r+1}))$ if and only if $\sigma$ preserves these colors in $A_r$. 
Therefore, the isomorphism  problem for ternary graphs is polynomial-time  reducible to the following color automorphism problem (taking $G=\Aut_e(X_{r})$, $B=A=A_r$, $\sigma=\mathrm{Id}$).

\begin{problem}\label{Problem4}
Given
\begin{itemize}
\item A set of generators for a $2$-subgroup $G$ of the symmetric group $\Sym(A)$ of a colored set $A$
\item A $G$-stable subset $B\subseteq A$ 
\item A permutation $\sigma \in Sym(A)$
\end{itemize}
find  $C_B( \sigma G)$, where, for every $T\subseteq \Sym(A)$
$$
C_B(T)= \{ \tau \in T\mid  \tau \mbox{ preserves the color of every } b \in B \} 
$$
\end{problem}

Now, the following three lemmas are the basis of Algorithm \ref{CB(G)}   that solves Problem \ref{Problem4} in polynomial-time, and thus it completes the graph isomorphism test we were looking for.

\begin{lemma}\label{LeftCoset}
 Let $G$ be a subgroup of $\Sym(A)$, $\sigma \in \Sym(A)$ and $B$ a $G$-stable subset of $A$ such that $C_B( \sigma G)$ is not empty. Then, $C_B(G)$ is a subgroup of $G$ and $C_B( \sigma G)$ it is a left coset of the subgroup $C_B(G)$.
\end{lemma}

\begin{lemma}[Furst-Hopcroft-Luks \cite{FHL}]\label{LemaFilter}
 Given a set of generators of a subgroup $G$ of a symmetric group, one can compute in polynomial time  a set of generators of any subgroup of $G$ that is known to have polynomially bounded index in $G$ and for which a polynomial-time membership test is available.
\end{lemma}

\begin{lemma}[Luks \cite{Luks83}]\label{LemaBlockSystem}
 Given a set of generators of a subgroup $G$ of a symmetric group and a $G$-orbit $B$, one can determine in polynomial time a minimal $G$-block system in $B$.
\end{lemma}

\begin{algorithm}[h] \label{CB(G)}\hypertarget{CB(G)}{}
\KwData{A coset $\sigma G \subseteq \Sym(A)$, where $A$ is a colored set and $G$ is a $2$-group, and a $G$-stable subset $B$ of $A$.}
 \KwResult{ $C_B( \sigma G)$}
 
\Begin{
\Case{$B = \{ b \} $}{\eIf{ $\sigma(b) \sim b$}{ $C_B ( \sigma G ) = \sigma G$}{ $C_B = \emptyset $}
     }
\Case{ $G$ is intransitive on $B$}{
    Let $B_1$ a nontrivial orbit \\
    $B_2 = B \setminus B_1$ \\
    $C_B ( \sigma G ) = C_{B_2} ( C_{B_1} ( \sigma G)) $ 
   }
\Case{ $G$ is transitive on $B$}{
      Let $\{B_1, B_2 \}$ a minimal $G$-block system \\
      Find the subgroup  $H$ of $G$ that stabilizes $B_1$ \\
      Let $\tau \in G \setminus H$ \\
      $C_B ( \sigma G) = C_{B_2} ( C_{B_1}  (\sigma H)) \cup C_{B_2} ( C_{B_1}  (\sigma \tau H))$ \\
       }

\Return{ $C_B( \sigma G)$}
}
\caption{$C_B (\sigma G ) $}
\end{algorithm}

\begin{theorem}\label{th:cost}
Algorithm \ref{IsotoAute} solves the ternary graph isomorphism in time $O(n^{10})$, with $n$ the size of the input graphs.
\end{theorem}

\begin{proof}
The cost of  computing \texttt{BuildX} applied to a pair of graphs of size $n$ and a pair of edges is in $O(n)$ time.
For every $r$,
a set of generators of $\ker \pi_r$ can be computed as explained in Lemma \ref{K} in $O(n^2)$ time, a set of generators of  $ \pi_r ( \Aut_e(X_{r+1}))$ can be computed, by reducing this computation to Problem  \ref{Problem4} and then using Algorithm \ref{CB(G)}, in $O(n^8)$ time, and a set of preimages in $\Aut_e(X_{r+1})$ of the latter can be computed in $O(n^2)$ time. Since there are at most $O(n)$ indices $r$, we conclude that the cost of  computing  a set of generators of $\Aut_e(X)$ by means of Algorithm \ref{Aute} is in $O(n^9)$. Since Algorithm \ref{IsotoAute} calls Algorithm \ref{Aute} $O(n)$ times, the final cost of Algorithm \ref{IsotoAute} is in $O(n^{10})$.
\end{proof}

\section{Implementation}
\subsection{Improvements}
In our implementation we have improved the efficiency of Luks' algorithm by:
\begin{enumerate}[(a)]
 \item Reducing the size of the set $A_r$.
 \item Representing the groups by means of smooth generating sequences.
 \item Precomputing the blocks.
 \item Running initial tests, to avoid trivial cases.
 \item Removing, for every $r$, the permutations that do not swap the two ``parts'' $X_1$ and $X_2$ of $X$.
\end{enumerate}
With these improvements, the cost of the computation decreases to $O(n^4)$. We explain now in some detail  the improvements (a)--(c), which were inspired by \cite{GHL83}. For more details, see \cite{MasterThesis}.

\begin{figure}\label{Triplets}\hypertarget{Triplets}{ }
\centering
\includegraphics[scale=0.25]{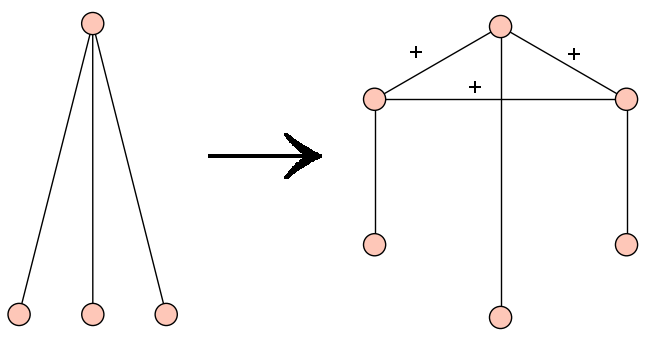}
\caption{Replacing the triplets in the neighbor sets}
\end{figure}

As far as improvement (a) goes, we have been able to remove the triplets from $B_r$, by replacing each node $v$ with a 3-elements neighbor set by a triangle with nodes at  ``level'' $r+1$ and labeled edges: cf. 
 \hyperlink{Triplets}{Figure \ref{Triplets}}. In this way, we can replace the graph $X$ by a new graph $\tilde{X}$ with some edges labeled. The automorphisms of $\tilde{X}$ must preserve labels, and therefore the computation of $\Aut_e(\tilde{X})$ is the same as that of $\Aut_e(X)$, except for the following facts: $B_r$ needs only to include the subsets of $V_r$ of size 1 or 2;  we split  $A'_r$  into 
$$
\begin{array}{l}
A'_{r,a}=  \big\{ \{v_1, v_2\} \in A_r\mid  ( v_1, v_2 ) \in  E(X_{r+1})\mbox{ is unlabeled} \big\},\\
A'_{r,b}=  \big\{ \{v_1, v_2\} \in A_r\mid  ( v_1, v_2 ) \in  E(X_{r+1})\mbox{ is labeled} \big\};
\end{array}
$$
and we modify the set of colors on $B_r$ to distinguish, on the one hand, whether or not an element is in $A'_{r,a}$, or $A'_{r,b}$, or neither, and, on the other hand, whether it is in $A_{r,1}$, or $A_{r,2}$, or neither.

As to (b), we represent 2-groups in a way that makes easier several key computations.  

\begin{definition}
Let $G$ be a 2-group generated by $\{g_1, \ldots, g_k \}$. The sequence $( g_1, \ldots, g_k )$ is a \emph{smooth 
 generating sequence} (\emph{SGS}, for short) for $G$ if $[ \langle g_1, \ldots, 
 g_i \rangle : \langle g_1, \ldots, 
 g_{i-1} \rangle] \leq 2$, for $i=1, \ldots, k$.
 \end{definition}

SGS are preserved by homomorphisms and liftings, and if we know a SGS
for a 2-group $G$, then it is easy construct an SGS for a subgroup $H$ of index 2  \cite{ GHL83}.

\begin{lemma}
Let $G$ be a 2-group, $( g_1, \ldots, g_k )$ a SGS for it, and $H$ a subgroup of index 2. Let $j = \min \{ i\mid  g_i \notin H \}$ and set, for $ i = 1, \ldots, k$,
 $$
\beta_i = \left\{ \begin{array}{lcl} g_i & \mbox{ if } & g_i \in H \\ g_j ^{-1} g_i & \mbox{ if } & g_i \notin H \end{array} 
  \right.
$$
Then $(\beta_1, \ldots, \beta_k )$ is a SGS for $H$, and this sequence is computed in $O(k)$ times the required time of a membership test for $H$.
\end{lemma}

Let us consider finally improvement (c). As we have seen in the proof of Theorem \ref{th:cost}, 
the most expensive part of Luks' algorithm is the recursive calls performed by Algorithm \ref{CB(G)}. The task carried out by this algorithm can be reorganized so as to limit the number of different blocks  visited. These blocks form a tree that is precomputed and guides the recursion.

\begin{definition}
Let $G$ be a 2-group acting on a stable subset $B$ of a colored set $A$. We call a binary tree $T$ a \emph{structure tree} for $B$ with respect to $G$, $T=T(B,G)$, if  the set of leaves of $T$ is $B$, and the action of any $\sigma \in G$ on $B$ can be lifted to an automorphism of $T$.
\end{definition}

Let $Q$ be a fixed color of $A$ (in our application, it will be the color of the elements of $A_r$ that do not belong to $A_r'\cup A_{r,1}\cup A_{r,2}$). A node $\tilde{B}$ of $T = T(B,G)$ will be called  \emph{active} if $\tilde{B} \cap Q^c \neq \emptyset$. An active node $\tilde{B}$ is \emph{facile} if $G$ is intransitive on $\tilde{B}$ and the latter has exactly one active child.  Let $\Delta ( \tilde{B} )$ denote the closest non facile descendant of $\tilde{B}$ in $T$.

We can precompute the entire structure tree $T(B,G)$ for a given pair  $(B,G)$ using Algorithm \ref{T(B,G))} below. With this algorithm,  we can construct every structure tree $T(B_r, G_r)$ in $O(n^2)$ time, and we can compute the mapping $\Delta$ and the set of active nodes in  $O(n \log n)$ time. 
Using structure trees to guide the recursion, the cost of computing $C_B(\sigma G)$ decreases to $O(n^4)$.

\begin{algorithm}\label{T(B,G))}\hypertarget{T(B,G)}{}
\KwData{$B, G$}
\KwResult{ $T=T(B,G)$}
\Begin{
Let the root of $T$ be $B$ \\
\If{$|B|=1$}{\Return{}}
Find the orbits of $G$ in $B$ \\
\If{$G$ is transitive}{
Find a minimal block system $\{ B_L, B_R \}$ for $G$ on $B$ \\
Find the subgroup $H$ of $G$ that stabilizes $B_L$ \\
Find $\tau \in G \setminus H$ \\
\Return{ $T = Tree(B_L,H) \cup \tau(Tree(B_L, H)$(linked to the new root $B$)}
}
\Else{Partition $B$ into two nontrivial $G$-stable subsets $B_L, B_R$ \\
\Return{$T= T(B_L, G) \cup T(B_R,G)$ (linked to the new root $B$)}}

}
\caption{$T(B,G)$}
\end{algorithm}

\subsection{Implementation details}
We have implemented Luks' algorithm using the language Python and some specific SAGE libraries for handling  graphs and groups of permutations.

Besides the obvious classes and functions necessary to implement the algorithm, we use a new class, called \texttt{Node}, for the structure tree. This class  has four  attributes:
\begin{itemize}
 \item \texttt{Node},   the content of the node.
 \item \texttt{Left},   the left child.
 \item \texttt{Right},  the right child.
 \item \texttt{Parent},   the parent.
\end{itemize}
When we build a new \texttt{Node} without  some attributes, they will be empty arrays. This class has the necessary functions  to modify its attributes, as well as  the following functions:
\begin{itemize}
 \item \texttt{IsLeaf()}, to know whether the node is a leaf.
 \item \texttt{Istransitive()}, to know whether $G$ is transitive on the node.
 \item \texttt{Isactive(Q)}, to know whether the node is active w.r.t. $Q$.
 \item \texttt{Isfacile(Q)}, to know whether the node is facile w.r.t. $Q$.
 \item \texttt{Delta(Q)}, to know the nearest non facile descendant of the node.
\end{itemize}

To test whether two graphs are isomorphic or not, we can use two different functions, \texttt{Isomorphism} and \texttt{Isomorphism2}. Both functions answer the question whether the input graphs are isomorphic, but moreover
the first function returns the whole group of automorphisms that fix the distinguished  edge $e$ of the graph  $X=\mathtt{BuildX}(X_1, X_2, e_1, e_2)$, while the second one only returns the subgroup of those automorphisms of $X$ that  swap the parts $X_1$ and $X_2$. 

Finally, we have adapted Luks' algorithm to test the isomorphism of  fully resolved rooted phylogenetic networks  on a given set of taxa. A \emph{rooted phylogenetic network}, on a set of taxa $S$ is a rooted, directed, acyclic graph with its leaves bijectively labeled in $S$. These graphs are used as explicit models of  evolutionary histories that, besides mutations, include reticulate evolutionary events like genetic recombinations, lateral gene transfers or hybridizations.  An evolutionary network is \emph{fully resolved} , or binary, when, for every node $v$ in it, the ordered pair $(d_{in}(v),d_{out}(v))$ is either $(0,2)$ (the root), $(1,0)$ (the leaves), $(1,2)$ (the \emph{tree nodes}) or $(2,1)$ (the \emph{reticulate nodes}). Two evolutionary networks on $S$ are \emph{isomorphic}  when they are isomorphic as directed graphs and the isomorphism preserves the leaves' labels. For more on phylogenetic networks, see \cite{Huson}. 

The main difference between the case of rooted phylogenetic networks and the general case is that, in the former, the isomorphisms map the root to the root, and therefore the graph $X=\mathtt{BuildX}(X_1, X_2)$ can be simply obtained by connecting the roots by an edge $e$. In particular,  Algorithm \ref{IsotoAute} needs not to  call  $O(n)$ times Algorithm \ref{Aute}, but only once, and the resulting cost is then in $O(n^3)$.

To test the isomorphism of fully resolved rooted phylogenetic networks we have defined the   function \texttt{IsomorphismPhilo}, which has the following parameters:
\begin{itemize}
 \item \texttt{X1}, the first phylogenetic network.
 \item \texttt{X2}, the second phylogenetic network.
 \item \texttt{n},  the number of nodes of X1.
 \item \texttt{dic1}, the dictionary with the labels of the nodes of the first graph.
 \item \texttt{dic2}, the dictionary with the labels of the nodes of the second graph.
 \item \texttt{r1},  the root of the first graph.
 \item \texttt{r2}, the root of the second graph.
\end{itemize}
\texttt{IsomorphismPhilo}  accepts networks with \emph{inner taxa} (with internal labeled nodes) as well as \emph{multilabeled networks}  (where different nodes can have the same label), although these features are not used (yet) in the phylogenetic networks literature. For the moment, the algorithm only accepts phylogenetic networks created with SAGE's function \texttt{Graph}, but in the near future we plan to adapt it so that it accepts networks described in the format eNewick \cite{eNewick}.

We have performed several tests on our implementation, which we report in the next subsection.
In the first and the third tests, we  have used SAGE functions  to generate random graphs and we have then made these graphs  connected and  trivalent, by randomly adding or removing edges when it was necessary. In the second test, we have used the function \texttt{graphs.DegreeSequence} that, given a sequence of degrees, returns a graph whose nodes have this sequence of degrees, if some exists. This has allowed us to know the  probability that the graphs under comparison were isomorphic.

To perform the tests on phylogenetic networks we defined some functions to create random fully resolved rooted phylogenetic networks:
\begin{itemize}
 \item\texttt{createDic(nodes,n)}, which, given a set of nodes, returns the following dictionary:
 \[
  d(i) = \left\lbrace \begin{array}{ll} \mbox{``\emph{i}''} & \mbox{ if }i \in nodes \\ \mbox{``\ ''} & \mbox{ if }i \notin nodes \end{array} \right.
 \]
 \item \texttt{createDic2(nodes,n)}, as above, but the selected nodes are labeled ``a'' or ``b''  equiprobably.
 \item \texttt{leaves(X1)}, returns the set of nodes of the network X1.
\end{itemize}
Then, we have developed an algorithm \texttt{RandomTree(n)} that returns a fully resolved rooted phylogenetic network with $n$ nodes. By default, each internal node has a probability of $0.5$  of being hybrid, but it can be changed by simply changing the probability parameter from $0.5$ to the desired probability.

The documentation of the whole module can be found in \url{http://www.alumnos.unican.es/aam35/sage-epydoc/index.html} and the code in \url{http://www.alumnos.unican.es/aam35/IsoTriGraph.py}.

\subsection{Tests}

The first two examples show that the  code works correctly, and then the tests  prove that it runs in a reasonable time, comparable to the speed of  the own SAGE algorithm to test isomorphisms of graphs.

\begin{example}\label{ImpTestEx1}
Consider the graphs $X_1$  and $X_2$ depicted in Fig. \ref{Example1}. These graphs are created in SAGE with:\begin{verbatim}
 sage: X3=Graph([(1, 7), (1, 10), (2, 3), (2, 4), (3, 4),(4, 9), 
       (5,6),(6, 8), (7, 8), (7, 9),(8, 9)])
 sage: X4=Graph([(2, 3), (2, 10), (1, 7), (1, 4), (7, 4),(4, 9), 
       (5, 6),(6, 8), (3, 8), (3, 9),(8, 9)])
\end{verbatim}

  \begin{figure}[h]
\centering
 \includegraphics[scale=0.4]{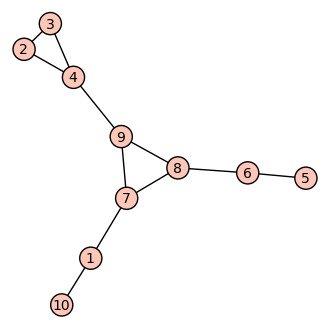}\qquad
\includegraphics[scale=0.4]{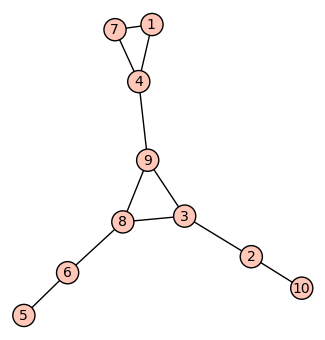}
 \caption{\label{Example1}
 The  graphs $X_1$ (left) and $X_2$ (right) in Example \ref{ImpTestEx1}}
\end{figure}

We test whether they are isomorphic: 
\begin{verbatim}
sage: Isomorphism2(X3,X4,10,Iso=True)
1 --> 2
2 --> 1
3 --> 7
4 --> 4
5 --> 5
6 --> 6
7 --> 3
8 --> 8
9 --> 9
10 --> 10
True
\end{verbatim}
And, indeed, it is obvious that this is an isomorphism between $X_1$ and $X_2$.
\end{example}

\begin{example}\label{ImpTestEx2}
Consider now the graphs $X_1$  and $X_2$ depicted in Fig. \ref{Example2}. 
\begin{verbatim}
sage: X1=Graph([(1, 7), (1, 8), (1, 10), (2, 3), (3, 6),
      (4, 5), (5, 6), (6, 10), (7,9), (7, 10), (8, 9)])
sage: X2=Graph([(1, 7), (1, 9), (2, 3), (2, 5), (2, 10), 
      (4, 5), (4, 6), (4, 10), (6,8), (7, 8), (7, 10)])
sage: Isomorphism(X1,X2,10)
False  
\end{verbatim}
And, indeed, they are clearly non isomorphic.

\begin{figure}[h!]
\centering
 \includegraphics[scale=0.4]{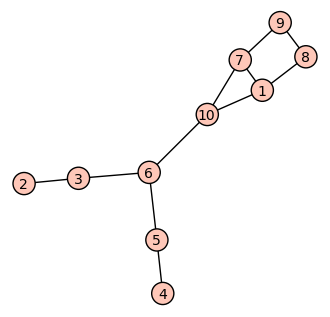}\qquad
\includegraphics[scale=0.4]{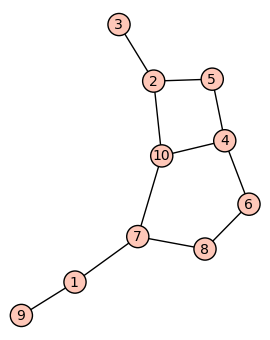}
 \caption{\label{Example2}
 The  graphs $X_1$ (left) and $X_2$ (right) in Example \ref{ImpTestEx2}}
\end{figure}
\end{example}

\begin{test}
Fig. \ref{Graphic1} shows the time needed by \texttt{Isomorphism2} to test whether two random graphs are isomorphic, as a function of the numbers of nodes in the graphs. The times are so small because two random graphs with the same number of nodes have  probably different numbers of edges, a property that our program checks before proceeding with Luks' algorithm. Although in most cases the algorithm detected non-isomorphism by trivial reasons, the algorithm is also relatively efficient in non trivial cases.
\begin{figure}[h!]
\centering
\includegraphics[scale=0.3]{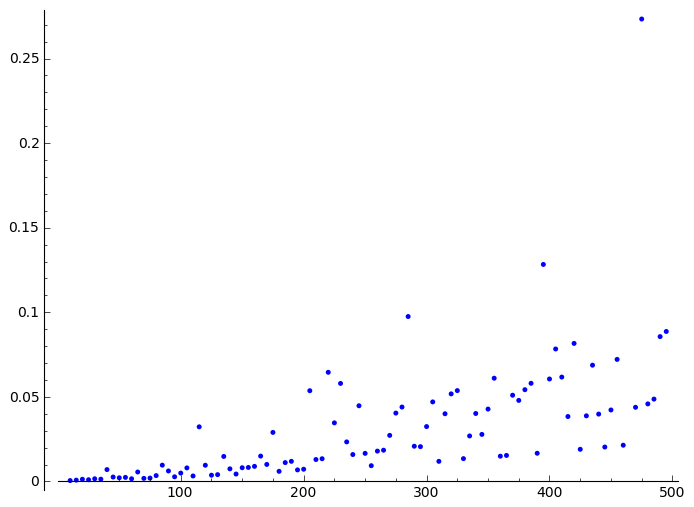}
\caption{Average time (in seconds) for different numbers of nodes when testing the isomorphism of random graphs}\label{Graphic1}
\end{figure}
\end{test}

\begin{test}
In this test we fix the degrees of $n-1$ nodes in the graphs, and the degree of the last node is chosen at random. In this way we guarantee that the probability of two graphs being isomorphic is $1/3$. In this case,  our algorithm also runs in reasonable time: see  Figs. \ref{Graphic2} and  \ref{Graphic3}.

\begin{figure}[h!]
\centering
\includegraphics[scale=0.3]{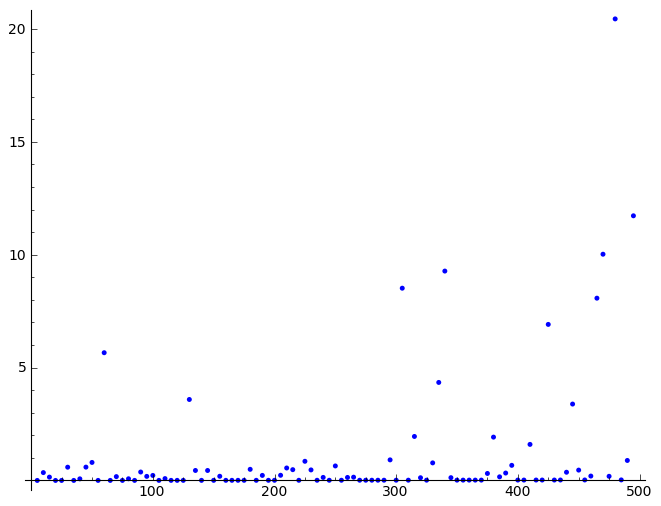}
\caption{Average time (in seconds) for different numbers of nodes when testing the isomorphism of semirandom graphs}\label{Graphic2}
\end{figure}

\begin{figure}[h!]
\centering
\includegraphics[scale=0.3]{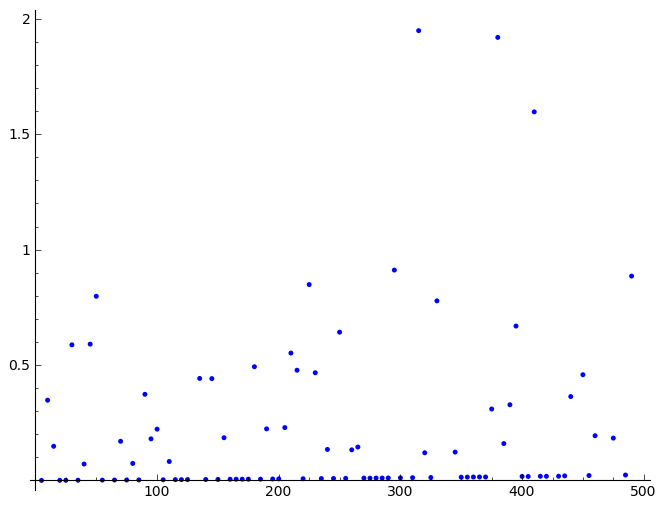}
\caption{The band under 2 seconds of the graphic in Fig. \ref{Graphic2}}\label{Graphic3}
\end{figure}
\end{test}

\begin{test}
Fig. \ref{Graphic4} shows the times (red dots) needed by the algorithm to detect the isomorphism between pairs of isomorphic graphs with $n$ nodes,  and it compares this time with  the  functions $(n/10)^4$, $(n/10)^3$, $(n/10)^2\log (n/10)$, and $(n/10)^2$.

\begin{figure}[h!]
\centering
\includegraphics[scale=0.3]{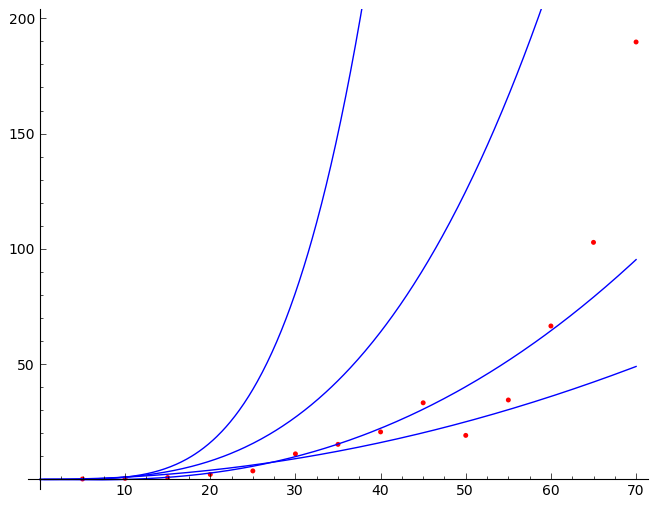}
\caption{Comparison between the algorithm and the functions $(n/10)^4, (n/10)^3, (n/10)^2 \log (n/10), (n/10)^2$}\label{Graphic4}
\end{figure}
\end{test}

Our last test deals with rooted phylogenetic networks.
\begin{test}
Fig. \ref{Graphic5} displays the relation running time-number of nodes for our algorithm when applied to random fully resolved rooted phylogenetic networks  on the same sets of taxa, and  Fig. \ref{Graphic6} shows this relation 
for phylogenetic networks that are isomorphic as undirected graphs but need not be isomorphic as phylogenetic networks.

\begin{figure}[h!]
 \centering
 \includegraphics[scale=0.3]{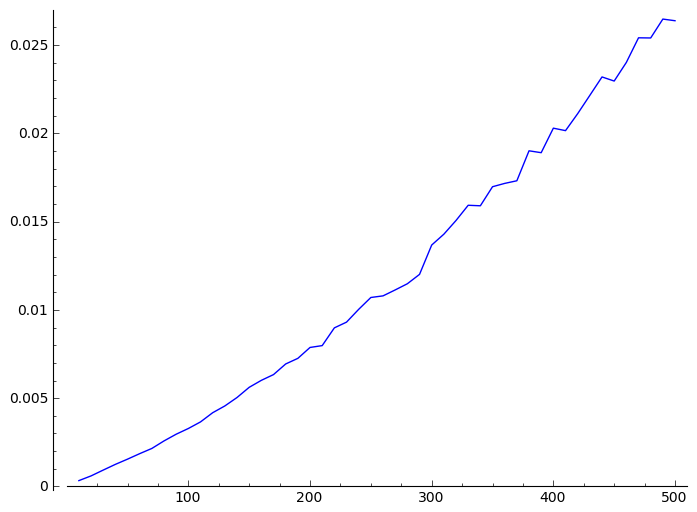}
\caption{Average time (in seconds) for different numbers of nodes when testing the isomorphism of fully resolved rooted phylogenetic networks}
\label{Graphic5}
 \end{figure}

\begin{figure}[h!]
 \centering
 \includegraphics[scale=0.3]{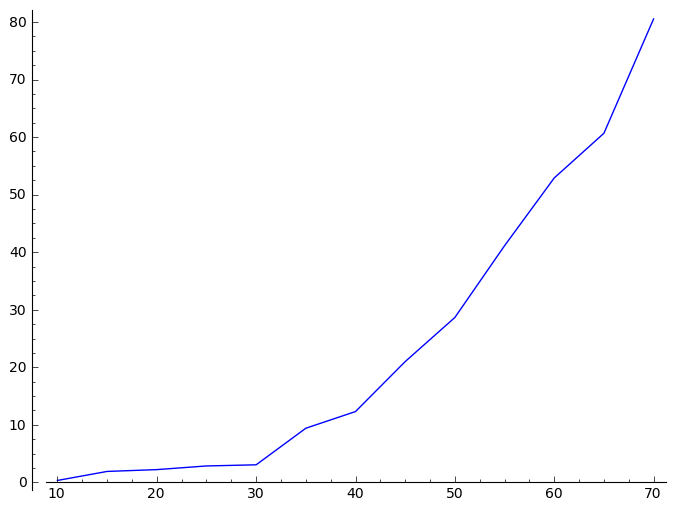}
\caption{Average time (in seconds) for different numbers of nodes when testing the isomorphism of fully resolved rooted phylogenetic networks  that are isomorphic as graphs}\label{Graphic6}
 \end{figure}
 \end{test}

\section{Conclusions}
In this paper we have presented  our implementation in SAGE of Luks' polynomial-time algorithm for testing the isomorphism of ternary graphs. This algorithm has interesting applications in phylogenetics,  as it allows, for instance, to detect whether two fully resolved rooted phylogenetic networks are isomorphic. Therefore, we have adapted and implemented Luks' algorithm for this type of graphs.  Our adaptation has been, except for one point, a direct translation of Luks' algorithm. Fully resolved rooted phylogenetic networks have specific characteristics that could be used to improve the algorithm to make it still more efficient in this specific application. It is in our research agenda to develop such an adaptation, and we hope to present it elsewhere. 
\medskip
 
\noindent\textbf{Acknowledgements.}  This  research has been partially supported by the Spanish government and the UE FEDER program, through project MTM2009-07165.  We thank Prof. T. Recio for encouraging us to write this report, and for his comments on a first version of it.

\end{document}